\documentclass[letterpaper,11pt]{article}
\usepackage{microtype}
\usepackage[letterpaper,margin=1in]{geometry}
\usepackage[english]{babel}
\usepackage[numbers,sort&compress]{natbib}
\usepackage{color}
\usepackage{doi}
\usepackage{latexsym}
\usepackage{amsmath}
\usepackage{amssymb}
\usepackage{amsthm}
\usepackage{thmtools}
\usepackage{thm-restate}
\usepackage[all,warning]{onlyamsmath}
\usepackage{enumitem}
\usepackage{framed}
\usepackage{xparse}
\usepackage{mathtools}
\usepackage{graphicx}
\usepackage[basic]{complexity}
\usepackage[small]{caption}
\usepackage{booktabs}
\usepackage{hyperref}
\usepackage{algpseudocode}

\theoremstyle{plain}
\newtheorem{theorem}{Theorem}
\newtheorem{lemma}[theorem]{Lemma}

\newtheorem{corollary}[theorem]{Corollary}
\newtheorem{claim}[theorem]{Claim}

\theoremstyle{definition}
\newtheorem{definition}[theorem]{Definition}

\theoremstyle{remark}

\setlist[itemize]{label=--}
\setlist[enumerate]{label=(\arabic*),labelindent=\parindent,leftmargin=*}

\DeclarePairedDelimiter\braces{\{}{\}}
\NewDocumentCommand\set{O{}mg}{\ensuremath{\braces[#1]{#2\IfNoValueTF{#3}{}{\,:\,#3}}}}

\DeclareMathOperator{\dist}{dist}

\DeclareMathOperator{\mypoly}{poly}

\DeclareMathOperator{\irr}{irr}
\DeclareMathOperator{\In}{in}
\DeclareMathOperator{\Out}{out}

\newclass{\lcl}{LCL}
\newclass{\LLL}{LLL}
\newclass{\local}{LOCAL}
\newclass{\dlocal}{DLOCAL}
\newclass{\rlocal}{RLOCAL}
\newclass{\LogStar}{LogStar}

\newcommand{\namedref}[2]{\hyperref[#2]{#1~\ref*{#2}}}

\newenvironment{myabstract}
{\list{}{\listparindent 1.5em%
		\itemindent    \listparindent
		\leftmargin    1cm
		\rightmargin   1cm
		\parsep        0pt}%
	\item\relax}
{\endlist}

\newenvironment{mycover}
{\list{}{\listparindent 0pt
		\itemindent    \listparindent
		\leftmargin    1cm
		\rightmargin   1cm
		\parsep        0pt}%
	\raggedright
	\item\relax}
{\endlist}

\newcommand{\myemail}[1]{\,$\cdot$\, {\small #1}}
\newcommand{\myaff}[1]{\,$\cdot$\, {\small #1}\par\medskip}

\hypersetup{
	colorlinks=true,
	linkcolor=black,
	citecolor=black,
	filecolor=black,
	urlcolor=[rgb]{0,0.1,0.5},
	pdftitle={Hardness of minimal symmetry breaking in distributed computing},
	pdfauthor={Alkida Balliu, Juho Hirvonen, Dennis Olivetti, Jukka Suomela}
}

\begin{document}

\begin{mycover}
	{\huge\bfseries\boldmath Hardness of minimal symmetry breaking in distributed computing \par}
	\bigskip
	\bigskip

	\textbf{Alkida Balliu}
	\myemail{alkida.balliu@aalto.fi}
	\myaff{Aalto University}

	\textbf{Juho Hirvonen}
	\myemail{juho.hirvonen@aalto.fi}
	\myaff{Aalto University}

	\textbf{Dennis Olivetti}
	\myemail{dennis.olivetti@aalto.fi}
	\myaff{Aalto University}

	\textbf{Jukka Suomela}
	\myemail{jukka.suomela@aalto.fi}
	\myaff{Aalto University}

\end{mycover}

\medskip
\begin{myabstract}
	\noindent\textbf{Abstract.}
	A graph is \emph{weakly $2$-colored} if the nodes are labeled with colors black and white such that each black node is adjacent to at least one white node and vice versa. In this work we study the distributed computational complexity of weak $2$-coloring in the standard \local{} model of distributed computing, and how it is related to the distributed computational complexity of other graph problems.

	First, we show that weak $2$-coloring is a \emph{minimal} distributed symmetry-breaking problem for regular even-degree trees and high-girth graphs: if there is any non-trivial \emph{locally checkable labeling problem} that is solvable in $o(\log^* n)$ rounds with a distributed graph algorithm in the middle of a regular even-degree tree, then weak $2$-coloring is also solvable in $o(\log^* n)$ rounds there.

	Second, we prove a \emph{tight lower bound} of $\Omega(\log^* n)$ for the distributed computational complexity of weak $2$-coloring in regular trees; previously only a lower bound of $\Omega(\log \log^* n)$ was known. By minimality, the same lower bound holds for any non-trivial locally checkable problem inside regular even-degree trees.
\end{myabstract}

\thispagestyle{empty}
\setcounter{page}{0}
\newpage


\section{Introduction} \label{sec:introduction}

In this work, we show that distributed symmetry-breaking problems in regular trees are either solvable in $O(1)$ rounds or they require $\Omega(\log^* n)$ rounds, and this is tight. This is an exponential improvement over the previous bound of $\Omega(\log \log^* n)$ rounds.

\paragraph{Symmetry breaking in distributed graph algorithms.}

Local symmetry breaking is a fundamental primitive in computer networks. As a simple example, consider the task of $3$-vertex coloring an $n$-cycle: the topology of the computer network is a cycle with $n$ nodes (computers) and $n$ edges (communication links), and each computer has to output a value from the set $\{1,2,3\}$ such that adjacent computers output different values.

If all nodes are identical and run the same deterministic algorithm, this is impossible: if all nodes start in the same state, they will send the same messages to their neighbors, they will receive the same messages from their neighbors, and switch to the same new state, ad infinitum -- and in $3$-coloring adjacent nodes have to reach distinct states. Hence we will need to assume that the nodes have e.g.\ \emph{unique identifiers} from a $\mypoly(n)$-sized set, or they have access to a source of \emph{random bits}.

With unique identifiers we can solve any graph problem trivially in $T = O(n)$ communication rounds, not only in cycles but more generally in any connected graph: all nodes can gather full information on the entire graph, solve the problem locally by brute force, and output their own part of the solution. We are interested in \emph{fast} distributed algorithms: each node should produce its own part of the output after some $T = o(n)$ communication rounds. It is good to note that here time (number of communication rounds) and distance (shortest-path distance in the input graph) are, in essence, equivalent: in $T$ rounds all nodes can gather full information on their $T$-round neighborhood, and nothing beyond that.

Distributed symmetry breaking in cycles is nowadays completely understood \cite{Naor1995,Brandt2017,chang17hierarchy}. For example, the task of $3$-coloring is solvable in time $T = \Theta(\log^* n)$ and this is tight; this holds for deterministic algorithms that use unique identifiers, as well as for randomized algorithms that work with high probability. We can also give a much stronger statement: Consider \emph{any} problem $P$ in which the task is to label the cycle with values from some constant-size set $X$, subject to some local constraints. Then $P$ falls in one of the following classes:
\begin{enumerate}[noitemsep]
\item \emph{trivial:} it is solvable in $O(1)$ rounds,
\item \emph{local:} its computational complexity is $\Theta(\log^* n)$ rounds,
\item \emph{global:} we need $\Theta(n)$ rounds just to check if a feasible solution exists.
\end{enumerate}
Examples of problems of type~(2) include vertex coloring with $k \ge 3$ colors, edge coloring with $k \ge 3$ colors, maximal independent set, and maximal matching. Examples of problems of type~(3) include vertex or edge coloring with $2$ colors.

However, once we switch from cycles to regular higher-degree graphs, we have a much more diverse landscape of computational complexity. The main gap of our current knowledge is in the low end of the complexity spectrum: we do not know if there are any problem of the above form that would require $\omega(1)$ rounds but that would be solvable in $o(\log^* n)$ rounds e.g.\ in regular high-girth graphs or regular trees \cite{chang17hierarchy}.

\paragraph{Prior work on distributed computational complexity.}

Formally, the model of computing that we study here is known as the $\local$ model of distributed computing \cite{Linial1992,Peleg2000}: the same graph plays the dual role of being both the unknown input graph and the structure of the computer network, the nodes correspond to computers, the edges correspond to communication links, the nodes are labeled with unique identifiers from a $\mypoly(n)$-sized set, computation proceeds in synchronous communication rounds, and the complexity measure that we care about is the number of rounds until all nodes have stopped and announced their local outputs. In particular, local computation is free and the message size is unbounded.

In the past three years, there has been a long sequence of papers \cite{Balliu2018stoc,Balliu2018disc,Brandt2016,chang16exponential,chang17hierarchy,fischer17sublogarithmic,ghaffari17distributed,Ghaffari2018,Ghaffari2018a,Pettie2018,Ghaffari2017} that have studied the computational complexity landscape of graph problems in the $\local$ model. The primary goal has been to understand so-called $\lcl$ problems, or \emph{locally checkable labeling} problems. In brief, these are problems in which the task is to label nodes with values from a constant-sized set, subject to local constraints. Informally, this is a natural distributed analogue of the class $\NP$: all $\lcl$ problems are such that, given a feasible solution, it is easy to \emph{verify} in a distributed manner. Now given an $\lcl$ problem $P$, what can we say about its computational complexity in the $\local$ model?

Surprisingly, the answer is that we can say quite a lot, even though the family of $\lcl$ problems is very diverse. For example, if we focus on bounded-degree graphs, we can prove \emph{gap results}: there is no $\lcl$ problem whose deterministic or randomized distributed complexity falls between $\omega(1)$ and $o(\log \log^* n)$ rounds, and there is no problem whose deterministic distributed complexity falls between $\omega(\log^* n)$ and $o(\log n)$ rounds.

\paragraph{Challenges below the log-star threshold.}

In the case of bounded-degree graphs, there are $\lcl$ problems with complexities such as $\Theta(\log \log^* n)$ and $\Theta(\sqrt{\log^* n})$ \cite{Balliu2018stoc}. However, all known problems in this region are highly artificial, and the worst-case input graphs have a rather peculiar structure, with lots of short cycles.

No \emph{natural} $\lcl$ problems in this region are known, and neither are there any known $\lcl$ problems that would have a computational complexity between $\omega(1)$ and $o(\log^* n)$ e.g.\ in bounded-degree \emph{trees} or regular \emph{high-girth graphs}. This is an open question identified by Chang and Pettie~\cite{chang17hierarchy}: do any such $\lcl$ problems exist at all?

The main reason that this question has remained open is related to the limitations of the current proof techniques. The proof technique introduced by Naor and Stockmeyer \cite{Naor1995} is based on the following idea: if there is an algorithm $A$ that solves an $\lcl$ problem $P$ such that each node only sees $o(\log^* n)$ other nodes in its radius-$T$ neighborhood, then using the hypergraph version of Ramsey's theorem we can construct another algorithm $A'$ that solves the same problem $P$ in the same running time in an \emph{order-invariant} manner: $A'$ does not use the numerical values of the identifiers but only their relative order.

If we apply this idea to a cycle, we can set $T = o(\log^* n)$. Then each node in $A$ makes its own decision based on the identifiers of $o(\log^* n)$ other nodes, and hence we can turn $A$ into order-invariant $A'$. And now $A'$ cannot solve any non-trivial $\lcl$ problem (consider the case in which the nodes are placed in an increasing order along the cycle; most of the nodes have then isomorphic radius-$T$ neighborhoods w.r.t.\ such a total order).

However, if we try to apply this idea to e.g. $4$-regular trees, we must have $T = o(\log \log^* n)$ in order to guarantee that radius-$T$ neighborhoods contain only $o(\log^* n)$ nodes. The technique does not say anything about e.g.\ algorithms that would have a running time of $T = O(\sqrt{\log^* n})$. Some ad-hoc extensions of the proof technique exist for the case of $d$-dimensional grids \cite{chang17hierarchy}, but nothing of that sort is known about trees.

\paragraph{\boldmath Focus on homogeneous $\lcl$s.}

$\lcl$s are a broad family of problems. At one extreme, we have got symmetry-breaking problems such as $3$-coloring of a cycle: the nodes have isomorphic local neighborhoods, yet they need to produce different outputs. At the other extreme, we have got problems in which nontrivial instances are exactly those in which symmetry has already been broken: for example, typical load balancing problems are trivial if the load is already uniformly distributed, and nontrivial only if there are differences in the load.

We focus in this work only on symmetry-breaking problems. In essence, we look at $\lcl$ problems that are nontrivial inside a regular even-degree tree (or regular even-degree high-girth graph), and trivial in all other cases. We formalize the notion of \emph{homogeneous} $\lcl$s in Section~\ref{sec:weak-is-minimal}. In brief, the idea is that we can turn any given $\lcl$ problem $P$ into a homogeneous $\lcl$ problem $P'$ such that, to solve $P'$, for each local neighborhood it is enough to either solve $P$ or find an irregularity in the input graph (for example, a short cycle or a low-degree node). Problem $P'$ is well-defined in any input graph, but the worst-case instances will be regular balanced trees.

It is easy to see that any homogeneous $\lcl$ is solvable in $O(\log n)$ rounds by brute force, and by prior work we know there are several examples of homogeneous $\lcl$s with deterministic complexities $\Theta(\log n)$ \cite{Brandt2016,chang16exponential,ghaffari17distributed,panconesi95delta}, $\Theta(\log^* n)$ \cite{Linial1992,cole86deterministic,Goldberg1988}, and $\Theta(1)$ \cite{Naor1995}. There is also a gap between $o(\log n)$ and $\omega(\log^* n)$ \cite{chang16exponential}, and another gap between $o(\log \log^* n)$ and $\omega(1)$ \cite{Naor1995,chang17hierarchy}. However, nothing is currently known about the region between $o(\log^* n)$ and $\Omega(\log \log^* n)$.

\paragraph{Contribution.}

We show that there are no homogeneous $\lcl$s with complexity between $o(\log^* n)$ and $\Omega(\log \log^* n)$. Together with prior work \cite{Naor1995,chang16exponential,chang17hierarchy,chang18complexity}, we have now a complete characterization of homogeneous $\lcl$s: each such problem falls in one of the classes listed in Table~\ref{tab:summary}, and all of the classes are nonempty.

In this work we focus on the proof of the gap between $o(\log^* n)$ and $\Omega(\log \log^* n)$, as everything else follows from prior work; we give an overview of the other gap results in Appendix~\ref{sec:gap-appendix}.

\begin{table}
\centering
\begin{tabular}{lll}
\toprule
Deterministic & Randomized & Example \\
\midrule
$\Theta(\log n)$ & $\Theta(\log n)$ & $2$-coloring \\
$\Theta(\log n)$ & $\Theta(\log \log n)$ & sinkless orientation \\
$\Theta(\log^* n)$ & $\Theta(\log^* n)$ & weak $2$-coloring in even-degree graphs \\
$\Theta(1)$ & $\Theta(1)$ & weak $2$-coloring in odd-degree graphs \\
\bottomrule
\end{tabular}
\caption{All possible distributed time complexities of homogeneous $\lcl$s; see Appendix~\ref{sec:gap-appendix} for details.}\label{tab:summary}
\end{table}

\paragraph{Approach, part 1: identifying a minimal problem.}

In complexity theory, the concept of \emph{complete} problems has played a fundamental role. In this work we introduce the complementary concept of \emph{minimal} problems. Informally:
\begin{itemize}[noitemsep]
    \item $P$ is $X$-complete: if we can solve $P$ efficiently, we can solve any $Q \in X$ efficiently.
    \item $P$ is $X$-minimal: if we can solve any $Q \in X$ efficiently, we can solve $P$ efficiently.
\end{itemize}
Let $\LogStar$ be the class of homogeneous $\lcl$s with time complexity $O(\log^* n)$ and $\Omega(\log \log^* n)$. In this work we:
\begin{itemize}[noitemsep]
    \item identify a problem $P$ that is $\LogStar$-minimal,
    \item prove a lower bound of $\Omega(\log^* n)$ for the time complexity of problem $P$.
\end{itemize}
Hence all problems in class $\LogStar$ will require $\Theta(\log^* n)$ rounds.

It turns out that the following problem is $\LogStar$-minimal: \emph{weak $2$-coloring in even-degree graphs}. In a weak $2$-coloring, the task is to label the nodes with colors $1$ and $2$ such that each node with label $1$ is adjacent to at least one node with label $2$ and vice versa. This is the problem studied in the seminal work by Naor and Stockmeyer \cite{Naor1995} -- they showed that weak $2$-coloring is solvable in $O(1)$ rounds in \emph{odd-degree} graphs. We prove a lower bound of $\Omega(\log^* n)$ rounds for \emph{even-degree} graphs.

The claim is trivial in the case of degree $2$ (i.e., cycles); this follows from the standard Ramsey-theoretic argument and also more directly from Linial's \cite{Linial1992} lower bound. The interesting cases are degrees $4, 6, \dotsc$, where one may think that, since it is necessary to break symmetry with just one neighbor, the problem becomes \emph{easier} as the degree increases. Note that, while for the standard $\Delta+1$ coloring problem we can prove a lower bound of $\Omega(\log^* n)$ on graphs of higher degree by reduction from the $2$-regular case, an algorithm for finding a weak $2$-coloring in e.g.\ $4$-regular graphs does \emph{not} imply an algorithm for solving the same problem in $2$-regular graphs.

We prove that weak $2$-coloring in regular trees of any even constant degree requires $\Omega(\log^* n)$ rounds, even for a randomized algorithm, and even if identifiers are exactly in $\{1,\ldots, n\}$.

\paragraph{Approach, part 2: lower bound for weak 2-coloring.}

Our proof technique is based on a \emph{speedup simulation argument}: On a high level, the idea is to show that given an algorithm $A_0$ that solves problem $P_0$ in time $T$, we can construct another algorithm $A_1$ that solves problem $P_1$ in time $T-1$, and eventually an algorithm $A_T$ that solves problem $P_T$ in time $0$. This is a contradiction, as $P_T$ is a nontrivial problem that cannot be solved in $0$ rounds. Hence $P_0$ cannot be solved in $T$ rounds, either.

In prior work, there have been two main types of such arguments. The first flavor is what Linial~\cite{Linial1992} and Naor~\cite{Naor1991} used (see also \cite{Laurinharju2014}). Here in each iteration we have the same problem---vertex coloring---but with different parameters. Here is an informal version of the main idea:
\begin{itemize}
	\item Assume that algorithm $A$ finds a $c$-vertex coloring in time $T$. Then we can construct algorithm $A'$ that finds a $2^c$-vertex coloring in time $T-1$.
\end{itemize}

The second flavor is what Brandt et al.~\cite{Brandt2016} used. Here we have two different problems---in their case so-called sinkless orientation and sinkless coloring---and we alternate between them. Here is an informal version of the main idea:
\begin{itemize}
	\item Assume that algorithm $A$ finds a sinkless orientation in time $T$. Then we can construct algorithm $A'$ that finds a sinkless coloring in time $T$.
	\item Assume that algorithm $A'$ finds a sinkless coloring in time $T$. Then we can construct algorithm $A''$ that finds a sinkless orientation in time $T-1$.
\end{itemize}
Here sinkless orientation is an \emph{edge-based problem}, in which the task is to label each edge, while sinkless coloring is a \emph{node-based problem}, in which the task is to label each node. By alternating between edge-based and node-based perspectives (and also edge-centric and node-centric models of distributed computing), we can turn any algorithm that finds a sinkless orientation into a faster algorithm that solves the same problem (we only lose in the success probability here).

In this work we combine the above two ideas. We alternate between a node-based and an edge-based problem, similar to what Brand et al.~\cite{Brandt2016} did. However, we replace a single node-based problem with a family of node-based problems $P_i$, in the spirit of Linial~\cite{Linial1992} and Naor~\cite{Naor1991}, and similarly we replace a single edge-based problem with a family of edge-based problems $Q_i$:
\begin{itemize}
	\item Assume that algorithm $A$ solves $Q_i$ in time $T$. Then we can construct algorithm $A'$ that solves $P_i$ in time $T$.
	\item Assume that algorithm $A'$ solves $P_i$ in time $T$. Then we can construct algorithm $A''$ that solves $Q_{i+1}$ in time $T-1$.
\end{itemize}
Iterating this process, we get a $0$-time algorithm for solving a nontrivial problem $P_T$ (with a sufficiently high success probability), which will be a contradiction.

In our case the node-based problem $P_i$ will be weak coloring with $f(i)$ colors, for a suitable choice of function $f(i)$. One of the main challenges is identifying a suitable edge-based problem $Q_i$ so that the recursion works; a somewhat unusual variant of non-proper edge coloring will do the trick.

\paragraph{Open questions for future work.}

The main question left open is to extend the result from homogeneous $\lcl$s to arbitrary $\lcl$s on trees or high-girth graphs. We conjecture that the same gap between $o(\log^* n)$ and $\omega(1)$ holds also in that case.

\section{Model and definitions}
We consider simple, undirected and connected graphs $G=(V,E)$. We denote by $n=|V|$ the size of the graph. The \emph{distance} between two nodes $u$ and $v$, $\dist(u,v)$, is the number of edges on the shortest path between $u$ and $v$. A \emph{labeling} of a graph $G$ is a mapping $\ell \colon V \to \Sigma$. If we consider a labeled graph $(G,\ell)$, the \emph{$t$-radius neighborhood} of a node $v$, $B_t(v)$, is a pair $(G',\ell')$, where $G'$ is the subgraph of $G$ induced by all nodes at distance at most $t$ from $v$, and $\ell'$ is the labeling function $\ell$ restricted to the nodes of $G'$. The $t$-radius neighborhood of an edge is the union of the $t$-radius neighborhoods of its endpoints.

\subsection{\boldmath The \local{} model}
We consider the standard \local{} model of distributed computing \cite{Peleg2000,Linial1992}. A distributed network is represented as a graph $G=(V,E)$, where each node $v \in V$ represents a computational entity, and an edge $e \in E$ represents a communication link between two entities.  Each node $v \in V$ runs the same algorithm $\mathcal{A}$. An algorithm is randomized if, on each node, it can access a private sequence of unbiased random bits, otherwise it is deterministic. The sequence of random bits can be seen as a labeling of the nodes, mapping a node to its own sequence of bits. The maximum degree of $G$ is bounded by some constant $\Delta$, and this parameter is known to $\mathcal{A}$. Also, $\mathcal{A}$ knows the size of the network $n=|V|$. Nodes may be provided with identifiers in $\{1,\ldots,n^c\}$, where $c \ge 1$ is a constant known to the algorithm. If the algorithm is not provided with identifiers, then $\mathcal{A}$ is said to be \emph{anonymous}. While in the standard \local{} model nodes may be provided with some additional input, for our purpose we assume that nodes do not have any.

Initially, each node $v \in V$ knows just its degree, the size of the network $n$ and the maximum degree $\Delta$. Then, the computation proceeds in synchronous rounds. At each round, each node $v \in V$ does the following:
\begin{itemize}[noitemsep]
	\item sends a (possibly different) message to each of its neighbors,
	\item receives the messages sent by the neighbors,
	\item performs some local computation.
\end{itemize}

After each round, a node may decide to halt and produce an output. We do not impose any constraint on the size of the messages exchanged by the nodes. Also, we do not impose constraints on the local computational power. The only parameter we are interested to analyze is the number of rounds required to solve a task. The running time of an algorithm is the number of rounds required until all nodes have stopped. Since there are no constraints on the bandwidth nor on the amount of local computation, we can see a $T$-round algorithm $\mathcal{A}$ as a mapping from radius-$T$ neighborhoods to valid outputs.

\subsection{Locally checkable labelings}
Locally checkable labellings (\lcl{}s) were introduced by Naor and Stockmeyer \cite{Naor1995}. These graph problems are described by a constant-size set of possible input labels $\Sigma_{\In}$, a constant-size set of possible output labels $\Sigma_{\Out}$, and a set of local constraints $\mathcal{C}$.  Given an \lcl{} $P = (\Sigma_{\In},\Sigma_{\Out},\mathcal{C})$, each node $v \in V$ is labeled with some input from $\Sigma_{\In}$, and the goal is to label all nodes with some output from $\Sigma_{\Out}$, such that each $r$-radius neighborhood of the graph satisfies the constraints $\mathcal{C}$. Here $r=O(1)$ is called the \emph{radius} of the \lcl{}. We assume that the input graphs have maximum degree at most $\Delta=O(1)$. Since $\Delta$ and $r$ are constants, the local constraints can be described as a finite set of valid $r$-radius input-output labeled neighborhoods. In this work, we consider \lcl{}s in which nodes have no input, that is, $\Sigma_{\In} = \{\bot\}$.

An example of \lcl{} is $(\Delta+1)$-coloring. In this case, $\Sigma_{\In} = \{\bot\}$, $\Sigma_{\Out} = \{1,\ldots,\Delta+1\}$, $r=1$, and the local constraints impose that a node of color $c \in \Sigma_{\Out}$ cannot be adjacent to another node of the same color $c$.

Let $\dlocal(T)$ and $\rlocal(T)$ be the classes of \lcl{}s that can be solved in the $\local$ model in time $T$ using a deterministic and a randomized algorithm, respectively.


\section{Minimal symmetry breaking} \label{sec:minimal}

Our plan is to show that, in a certain formal sense, weak colorings form a family of minimal symmetry breaking problems. We start with the observation that given a distance-$k$ weak $c$-coloring, for any constant $k$ and $c$, it is possible to compute a weak $2$-coloring in constant time. It then follows that if we consider any $\lcl$ $P$ that is nontrivial inside regular even-degree trees, a solution for $P$ also gives a solution for weak $2$-coloring.

\subsection{Weak colorings are constant-time reducible to each other}

Formally we consider the following family of weak colorings.
\begin{definition}
  For $c,k = O(1)$, a labeling $\varphi \colon V \to [c]$ is a distance-$k$ weak $c$-coloring if for every node $v$, there exists $u$ with $\dist(v,u) \leq k$ and $\varphi(v) \neq \varphi(u)$.
  A distance-$1$ weak $c$-coloring is also called a weak $c$-coloring.
\end{definition}

\begin{lemma} \label{lem:weak-col-family}
  Given a distance-$k$ weak $c$-coloring, for $k, c = O(1)$, it is possible to compute a weak $2$-coloring in time $O(1)$.
\end{lemma}

\begin{proof}
  First we compute a weak $2c$-coloring as follows. Let $\varphi$ denote the initial distance-$k$ weak $c$-coloring. Each node $v$ finds the closest node $u$ such that $u$ has a different color $\varphi(u) \neq \varphi(v)$ from $v$. Break ties by choosing one of the nodes with the smallest such color. Then $v$ outputs the color $\varphi'(v) = (\varphi(v), \dist(u,v) \bmod 2)$. The number of colors in $\varphi'$ is $2c$, and $\varphi'$ can be computed in $k$ rounds.

  For node $v$, if there is a neighbor $u$ such that $\varphi(v) \neq \varphi(u)$, then $\varphi'(v) \neq \varphi'(u)$ since the first elements are the old colors. Now assume that this was not the case and consider $w$, a neighbor of $v$ on the path to $u$, the closest node with a different color from $v$. Now $u$ must be the closest node of different color to $w$, since if there was a $u' \neq u$ with $\dist(w,u') < \dist(w,u)$ and $\varphi(w) \neq \varphi(u')$, then $u'$ would be the closest node to $v$ with a different color as well. Since $\dist(v,u) = \dist(w,u)+1$ we have that $\varphi'(v) \neq \varphi'(w)$.

  The last step is to reduce the number of colors from $2c$ to $2$, and here we can directly apply the algorithm by Naor and Stockmeyer \cite{Naor1995}. First use the standard Cole--Vishkin color reduction algorithm~\cite{cole86deterministic,Goldberg1988,Barenboim2013} to go from $2c$ colors to $3$ colors in time $O(\log^* c)$: each node picks one of its neighbors with a different color, and the color reduction algorithm is run on the resulting oriented pseudoforest~\cite[Section 5.2]{rybicki11exact}. This reduces the number of colors to 3. Finally, we use the 3-coloring to greedily find a maximal independent set in the pseudoforest and interpret this as a weak 2-coloring in the natural way.
\end{proof}

\subsection{Minimality of weak colorings}\label{sec:weak-is-minimal}

We start by defining a \emph{pointer problem} $P^*$, which we will later use to formalize the concept of \emph{homogeneous} $\lcl$s. Assume that $\Delta \in \{4,6,\dotsc\}$. In problem $P^*$, each node $v$ outputs a natural number $0 \le d(v) < \Delta$ and a possibly empty \emph{pointer} $p(v)$ that points to one of its neighbors. We say that a node $v$ is \emph{$P^*$-happy} if all of the following holds:
\begin{enumerate}[noitemsep]
  \item If $\deg(v) = \Delta$, then $p(v) = u$, where $u$ is one of the neighbors of $v$.
  \item If $\deg(v) < \Delta$, then $p(v) = \bot$ and $d(v) = \deg(v)$.
  \item If $p(v) = u$, then $d(v) = d(u)$ (pointer chains are labeled consistently).
  \item If $p(v) = u$, then $p(u) \neq v$ (pointer chains do not backtrack).
  \item If $p(v) = u$, then $p(u) \neq \bot$ or $\deg(u) = d(v)$ (pointer chains terminate at correct degree).
\end{enumerate}
We have a feasible solution for $P^*$ if all nodes are $P^*$-happy. Note that $P^*$-happiness can be verified locally in 1 round. If we encode the pointers as port numbers, we will have constantly many possible output labels for any constant $\Delta$, and hence $P^*$ is an $\lcl$ problem.

The definition of $P^*$ implies that there is a chain of pointers starting from each node $v$ of degree $\Delta$, consistently labeled with some $d < \Delta$, and that chain can only end in a node of degree $d < \Delta$, or there is a cycle in the pointer chain.

We call nodes of degree less than $\Delta$ and cycles of nodes of degree $\Delta$ \emph{irregularities}. We are interested in the closest irregularity to a given node. The distance $\dist(v,C)$ to a cycle $C$ is defined as $\min_{u \in C} \{ \dist(v,u) \} + \ell(C)$, where $\ell(C)$ is $|C|/2$ for even cycles and $\lfloor|C|/2\rfloor+1$ for odd cycles. This ensures that if a node does not have any irregularities within distance $r$, then its $r$-neighborhood is a regular tree.

The following lemma states that the 1-neighborhood of each node within distance $r$ to an irregularity can be feasibly labeled according to $P^*$ in time $O(r)$; we postpone the proof of this lemma to Section~\ref{ssec:homog-solving}.

\begin{restatable}{lemma}{lemhomoglocal} \label{lem:homog-local}
  Let $G$ be a graph of maximum degree $\Delta$. Problem $P^*$ can be solved in the $1$-neighborhood of all nodes with irregularities within distance $r$ in time $O(r)$.
\end{restatable}

Lemma~\ref{lem:homog-local} yields an upper bound of $O(\log n)$ for $P^*$, as there is always an irregularity at such a distance. Then in Section~\ref{ssec:homog-complexity} we prove a matching lower bound, which gives the following theorem:

\begin{restatable}{theorem}{thmsymcomplexity} \label{thm:symcomplexity}
  The distributed time complexity of $P^*$ is $\Theta(\log_{\Delta} n)$.
\end{restatable}

\paragraph{\boldmath Homogeneous \lcl{}s.}
In a homogeneous $\lcl$, we will need to solve some problem $P$ in those parts of the graph in which we do not have any irregularities, but we do not care about the correctness in the neighborhoods in which we have irregularities. We formalize this idea by constructing a homogeneous version $P_H$ of $P$, in which any node can always fall back to solving the pointer problem~$P^*$. Hence whenever a node sees any irregularities nearby, it can simply choose to construct a pointer chain pointing to an irregularity.

Recall that we assumed that $\Delta \in \{4,6,\dotsc\}$. Given an \lcl{} $P$ with maximum degree $\Delta$, we form a \emph{$\Delta$-homogeneous} \lcl{} $P_H = (P, P^*)$, in which the labeling $\ell_H$ is a pair $(\ell_P, \ell_{P^*})$, where we also allow $\ell_{P^*}(v)$ to be empty. At each node $v$ the verifier for $P_H$ accepts if and only if one of the following is true:
\begin{enumerate}[noitemsep]
  \item $v$ has a nonempty $P^*$-label and $v$ is $P^*$-happy, or
  \item $v$ has an empty $P^*$-label and verifier for $P$ accepts.
\end{enumerate}
In particular, if an algorithm starts solving $P^*$ in some neighborhood, the pointer chains it creates cannot terminate without meeting an irregularity.

\paragraph{\boldmath Classification of homogeneous \lcl{}s.}
Our aim is to characterize completely the computational complexity of homogeneous $\lcl$s. We will proceed as follows: In Sections \ref{sec:lboverview}--\ref{sec:generalization} we study the complexity of weak $2$-coloring and prove a lower bound of $\Omega(\log^* n)$ for it. Then in Section~\ref{sec:homog-proofs} we return back to the topic of homogeneous $\lcl$s. In Section~\ref{ssec:homog-classfication} we use Lemma~\ref{lem:weak-col-family} to show that homogeneous weak $2$-coloring is a \emph{minimal} nontrivial homogeneous problem, and hence there cannot be any homogeneous $\lcl$ with a complexity in the range between $\omega(1)$ and $o(\log^* n)$. Together with the previous work on the complexity theory of \lcl{} problems (see Appendix~\ref{sec:gap-appendix} for details), we will then obtain the following theorem:

\begin{restatable}{theorem}{thmhomogcomplexity} \label{thm:homog-complexity}
  A $\Delta$-homogeneous \lcl{} $P_H = (P, P^*)$ can have the following complexities.
  \begin{enumerate}[noitemsep]
    \item $O(1)$ deterministic and randomized,
    \item $\Theta(\log^* n)$ deterministic and randomized,
    \item $\Theta(\log n)$ deterministic and $\Theta(\log \log n)$ randomized, or
    \item $\Theta(\log n)$ deterministic and randomized.
  \end{enumerate}
\end{restatable}


\section{Lower bound for weak 2-coloring: overview}\label{sec:lboverview}

In Sections \ref{sec:speedup}--\ref{sec:generalization} we will give a lower bound of $\Omega(\log^* n)$ for weak 2-coloring. Specifically, we will prove the following theorem that holds even if the input graph is promised to be a regular tree of any even constant degree, also for randomized algorithms, even if a globally consistent orientation is provided, and even if nodes are provided with identifiers exactly in $\{1,\ldots,n\}$.

\begin{restatable}{theorem}{generalizedresult}\label{thm:generalizedresult}
	Solving weak 2-coloring on trees of any even constant degree with global success probability at least $\frac{1}{2}$ requires $\Omega(\log^* n)$ communication rounds.
\end{restatable}

We will first consider the case in which the tree is $4$-regular. In Section \ref{sec:speedup} we will prove a simulation result, informally stating that given an algorithm that can compute a weak $c$-coloring in $t$ rounds, we can find a weak $c'$-coloring in $t-1$ rounds, with bounded growth in error probability and palette size. Then, in Section \ref{sec:weak2collb}, we show how to use the simulation result to prove our main theorem for the case of $4$-regular trees. These proofs are heavily based on the ideas presented by Naor \cite{Naor1991}.

Finally, in Section \ref{sec:generalization}, we will show how our results can be generalized to any regular tree of constant degree, completing the proof of Theorem \ref{thm:generalizedresult}.

\section{Speedup simulation in 4-regular trees} \label{sec:speedup}

In this section we show that, given an algorithm for weak $c$-coloring on $4$-regular oriented trees, we can speed it up and obtain a weak coloring algorithm that is one round faster but uses a larger palette and has a larger local failure probability. The size of the new palette is doubly exponential in $c$, and the new algorithm has a failure probability that is polynomial in the failure probability of the original algorithm.

\paragraph{Setting: consistently oriented trees.}

We assume that nodes do not have unique identifiers, but have access to random bits. Nodes are assumed to know $n$, the size of the input graph. We also assume that the edges are oriented in a consistent manner, that is, a labeling in $\{U,D,L,R\}$, representing the directions \emph{up}, \emph{down}, \emph{left}, and \emph{right}, satisfying that if the edge $\{u,v\}$ is labeled $R$ for $u$, then it is labeled $L$ for $v$, and that if the edge is labeled $U$ for $u$, then it is labeled $D$ for $v$.

Since the graph topology is always assumed to be a 4-regular tree, $t$-round algorithm $A$ is a function that maps each $t$-neighborhood $B_t(v)$ to an output. The output $A(v)$ only depends on the random bit assignment to the neighborhood. An algorithm $A$ fails locally with probability at most $p$ if, for each node $v$, $\Pr[\,\forall u \in N(v), A(u) = A(v)\,] \le p$.

\paragraph{Intermediate problem: weak edge coloring.}

We will not directly show how to convert a $t$-round algorithm to a $(t-1)$-round algorithm. Instead, we define an intermediate problem that requires to produce an edge coloring such that for each node, either the edges labeled $U$ and $D$ have different colors, or the edges labeled $L$ and $R$ have different colors. If the number of allowed colors is $c$, we refer to this problem as \emph{weak edge $c$-coloring}. For this problem we construct an algorithm in a different, edge-based, model of computation: edges are computing entities, and two edges can communicate if they are incident to the same node. A $t$-round edge algorithm $A'$ is a function that maps $t$-edge neighborhoods $B_t(\{u,v\}) = B_{t}(u) \cup B_{t}(v)$, labeled with random bits, to the outputs. Notice that, even if it is an edge-based model of computation, random bits are still given to the nodes.
Since the topology is assumed to be a 4-regular tree, the output $A'(e)$ only depends on the type of the edge $e$ and the random bits. We say that an edge algorithm $A'$ fails locally around node $v$ with probability at most $p$ if $\Pr[ A'(e_U(v)) = A'(e_D(v)) \text{ and }  A'(e_L(v)) = A'(e_R(v)) ] \le p$, where $e_x(v)$ is the edge incident to $v$ labeled $x$.

We show how to convert a $t$-round node based algorithm into a $(t-1)$-round edge based algorithm in Lemma~\ref{lem:speedup1}. Then, we show how to convert a $(t-1)$-round edge algorithm into a $(t-1)$-round node algorithm in Lemma~\ref{lem:speedup2}.

\begin{lemma}[First speedup lemma] \label{lem:speedup1}
  Let $A$ be a $t$-round weak $c$-coloring algorithm with local failure probability at most $p$. Then there exists a $(t-1)$-round edge based algorithm $A'$ for weak edge $2^{2c}$-coloring with local failure probability $p' \leq 5 p^{1/5} \cdot c^{4/5}$.
\end{lemma}

\begin{lemma}[Second speedup lemma] \label{lem:speedup2}
  Let $A$ be a $t$-round weak edge $c$-coloring algorithm with local failure probability at most $p$. Then there exists a $t$-round algorithm $A'$ for weak $2^{4 c}$-coloring with local failure probability $p' \leq 4p^{1/4} \cdot c^{3/4}$.
\end{lemma}

\subsection{First speedup lemma} \label{ssec:speedup1proof}

\begin{proof}[Proof of Lemma~\ref{lem:speedup1}]
  We will construct the edge based algorithm $A'$ using a simulation of $A$. The $(t-1)$-neighborhood of an edge $e = (v,u)$, $B_{t-1}(e)$, almost fixes the $t$-radius neighborhoods of nodes $u$ and $v$. The part that is still not fixed, for nodes $u$ and $v$, is the random assignments of nodes in $B_t(v) \setminus B_{t-1}(e)$ and in $B_t(u) \setminus B_{t-1}(e)$. Note that these two sets are disjoint, since the topology is assumed to be a tree.

  The high level idea of the simulation is the following. Edge $e$ knows a part of the radius-$t$ neighborhood of nodes $u$ and $v$, and it can complete these partial views to full $t$-radius views in all possible ways. Each time it can simulate the original algorithm $A$ and see what color this algorithm would have assigned to the nodes. Thus, for each color $c$ and endpoint $v$, the edge $e$ knows the frequency with which the algorithm gives color $c$ to node $v$. This procedure gives $2c$ frequencies ($c$~for each endpoint). We can then apply a threshold function, obtaining a bit string of length $2c$, that we will show to be a weak edge $2^{2c}$-coloring.

  Thus, algorithm $A'$ does the following. Each edge $e$ starts by gathering $B_{t-1}(e)$, and then goes over all random bit assignments in $B_t(v) \setminus B_{t-1}(e)$ and in $B_t(u) \setminus B_{t-1}(e)$, simulates $A$ on $B_t(v)$ and $B_t(u)$, and computes the frequency of each color $i \in [c]$ appearing as the output of $v$ and $u$, respectively. Let $f \in [0,1]$ be a parameter to be decided later. We say that a color is \emph{frequent} for $v$ given $B_{t-1}(e)$, if it is the output of $A(v)$ with probability at least $f$.

  Let us fix an orientation for the nodes $u$ and $v$: edge $e$ is either $e_U(v)$ or $e_R(v)$ for $v$, while $e$ is either $e_D(u)$ or $e_L(u)$ for $u$.
  The output of $A'$ is a bit vector of length $2c$: for the first $c$ bits, the bit $i$ is $1$ if and only if the color $i$ is frequent for $v$. For the next $c$ bits, bit $c+i$ is $1$ if and only if color $i$ is frequent for $u$. Given an edge $e = (v,u)$, the output $A'(e)$ is also seen as the pair $(A'(v,e),A'(u,e))$ that consists of the sets of frequent colors $A'(v,e)$ of $v$ and $A'(u,e)$ of $u$.

  We will show that, for most of the random bit assignments, a local failure of algorithm $A'$ will also imply a local failure of the original algorithm $A$ with a smaller probability. We say that the assignment of random bits to $B_t(v)$ is \emph{good} if, for each edge $e$ incident to $v$, we have that $A(v)$, the true output of $A$, is contained in $A'(v,e)$. Since there are at most $c$ colors, the probability of having $A(v) \notin A'(v,e)$ is bounded by $cf$. By taking a union bound over the edges of $v$, we have that $\Pr[B_t(v) \text{ is good}] \geq 1 - 4cf$.

  \begin{figure}
    \centering
    \vspace{-4mm}
    \includegraphics[width=0.9\textwidth]{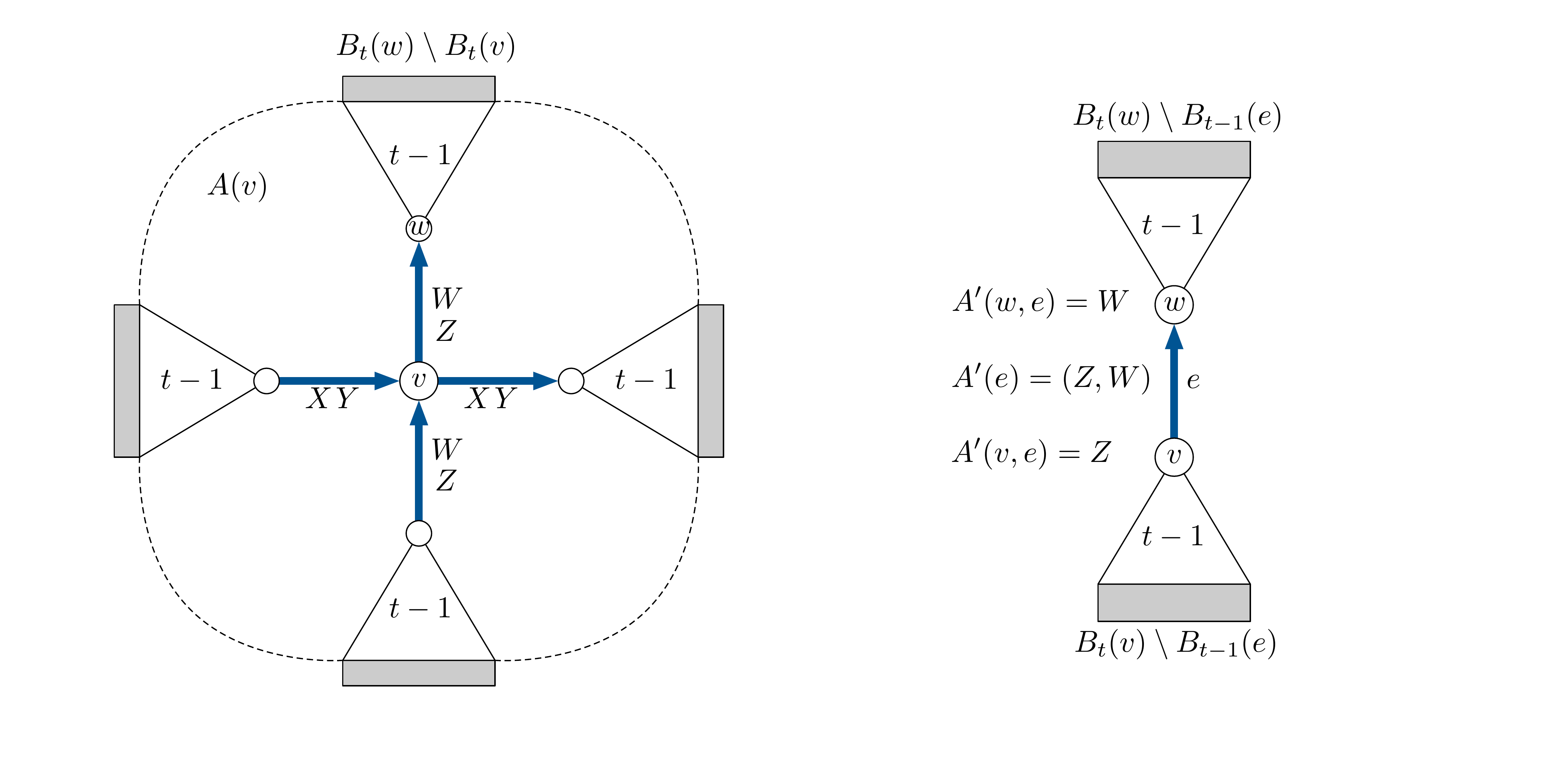}
    \vspace{-8mm}
    \caption{First speedup lemma: simulation of $A$ based on $(t-1)$-edge neighborhoods. On the left: the output of node $v$ is determined by the random bit assignment in its $t$-neighborhood. Given $B_t(v)$, the outputs of its neighbors are independent of each other. Assignment to $B_t(v)$ is good, if $A(v) \in X,Y,Z$, and $W$. On the right: the outputs of edges are constructed as simulations of $A$ at the endpoints $v$ and $w$. A color $i$ is included in the set $A'(w,e) = W$ if $w$ outputs $A(w)=i$ with probability $\geq f$ given $B_{t-1}(e)$. Given $B_{t-1}(e)$, the outputs of $v$ and $w$ are independent. Simulation $A'$ fails at $v$ if horizontal and vertical edges produce the same pairs of labels. In a good neighborhood this implies that $A$ also fails with probability at least $f^4$.}
    \label{fig:speedup1}
  \end{figure}

  Now assume that $A'$ fails locally in a good neighborhood $B_t(v)$. Let $e_x = e_x(v)$. Denote the partial outputs of edges $e_R, e_L, e_U$, and $e_D$ corresponding to node $v$ by $X, Y, Z$, and $W$, respectively. See Figure~\ref{fig:speedup1} for an illustration. Since by assumption $A'$ does not break symmetry neither horizontally nor vertically, we must have that $A'(e_R) = A'(e_L) = (X,Y)$ and that $A'(e_U) = A'(e_D) = (Z,W)$. Since the random bit assignment was good, we have that $A(v) \in X \cap Y \cap Z \cap W$.

  We have only fixed the random bits in $B_t(v)$. Since $A(v)$ is a frequent color for each neighbor of $v$ over their respective edges, this implies that for each neighbor $u$, $\Pr[A(u) = A(v) ~|~ B_t(v)] \geq f$, and that, since these events are independent, all neighbors of $v$ output $A(v)$ with probability at least $f^4$. Now if $A'$ fails locally with probability $p'$, we get that
  \begin{align*}
    p = \Pr[A \text{ fails locally}] &\geq \Pr[A' \text{ fails locally in a good neighborhood}] \cdot f^4 \\
    & \geq (1 -\Pr[A' \text{ does not fail}] - \Pr[\text{the neighborhood is not good}] )\cdot f^4 \\
    &\geq (p' - 4cf) \cdot f^4.
  \end{align*}
  To maximize this probability we set $f = p'/(5c)$. This implies that $p \geq (p' - 4cf)f^4 = (p'/5)^5 c^{-4}$, and conversely that $p' \leq 5 p^{1/5} \cdot c^{4/5}$.
\end{proof}

\subsection{Second speedup lemma} \label{ssec:speedup2proof}

\begin{proof}[Proof of Lemma~\ref{lem:speedup2}]
  We prove the statement for a $(t-1)$-round algorithm. We will again construct $A'$ by a simulation of $A$ and show that the local failure probability of $A'$ is bounded by the local failure probability of $A$.

  The algorithm $A'$ is constructed as follows. Given the random bits assigned to $B_{t-1}(v)$, node $v$ simulates $A$ on each of its incident edges. For each edge $e$, node $v$ outputs a bit vector of length $c$: the $i$-th bit is 1 if and only if $\Pr[A(e) = i ~|~ B_{t-1}(v)] \geq f$. Denote the vector of the edge $e$ by $A_e'(v)$. We say that the edge color $i$ is frequent for edge $e$ at node $v$ if $A_e'(v)$ contains a $1$ in position $i$. The output $A'(v)$ is then defined as the ordered 4-tuple consisting of the outputs of the right, left, up, and down edges of $v$ at $v$. Clearly, the output of $A'$ can take $2^{4c}$ different values.

  Now we bound the probability that the output of $A'$ is not a weak coloring. We will show that, if the algorithm $A'$ fails on some node $v'$, then $A$ fails on a neighbor $v$ of $v'$ with some probability.

  Let $r_{v'}$ and $u_{v'}$ denote the right and up neighbors of some node $v'$, respectively (see Figure \ref{fig:speedup2}). The $t$-neighborhood of node $v'$ is good if the following holds. For $(v',r_{v'})$ the output $A((v',r_{v'})) = x$ is frequent at node $v'$, and the output $A((v',u_{v'})) = y$ is frequent at $v'$ and $u_{v'}$. For an arbitrary edge $e$, the probability that $A(e)$ is not a frequent color of $e$ at its endpoint is, by definition, at most $cf$. By a union bound, the $t$-neighborhood of a node is good with probability at least $1-3cf$.

  \begin{figure}
	\centering
  \vspace{-3mm}
	\includegraphics[width=0.5\textwidth]{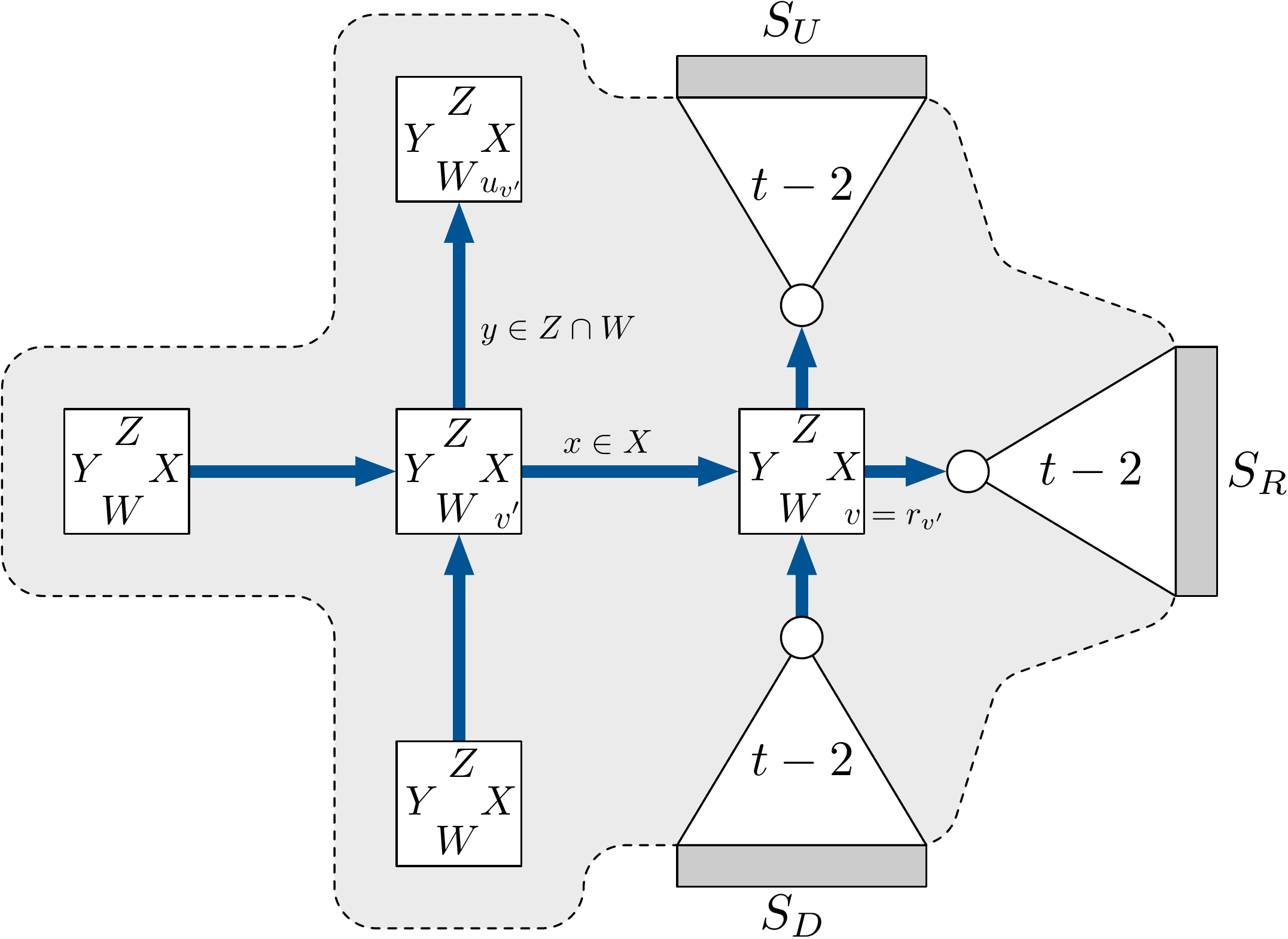}
	\caption{Second speedup lemma: bounding the local failure probability of $A$ based on $A'$. The neighborhood of node $v'$ is good if for its right and up edges $e$ and $e'$, the outputs $A(e) = x$ and $A(e') = y$ of those edges are in sets generated by the simulations in their endpoints: $x \in X$ and $y \in Z,W$. The simulation $A'$ fails at $v'$ if all neighbors of $v'$ also output the same 4-tuple $(X,Y,Z,W)$. This implies that, since the outputs of the edges incident to $v$ are independent given the $(t-1)$-neighborhood of $v$, the original algorithm $A$ fails around $v$ with probability at least $f^3$.}
	\label{fig:speedup2}
  \end{figure}

  Now consider arbitrary nodes $v,v'$ such that $v = r_{v'}$, and assume that the $t$-neighborhood around a node $v'$ is good, but that $A'$ fails at $v'$. Denote the output $A'(v')$ of $v'$ by $(X,Y,Z,W)$. The failure of $A'$ implies that $A'$ outputs $(X,Y,Z,W)$ for each neighboring node of $v'$ as well (including $v$). Since the neighborhood of $v'$ is good, we have that $x = A((v',v)) \in X$ and that $y = A((v',u_{v'})) \in Z \cap W$. Now consider $v$: given $B_{t-1}(v) \subseteq B_t(v')$, the outputs of its incident edges in $A$, not counting $(v',v)$, depend only on the random bit assignment in $B_t(v) \setminus B_t(v')$ (the regions $S_U, S_R, S_D$ in Figure \ref{fig:speedup2}), and the outputs are independent of each other given $B_{t-1}(v)$. Since $A((v',v)) = x$ and since we have that $x \in X$ and $y \in Z \cap W$, with probability at least $f^3$ the algorithm $A$ outputs $x$ on the right edge of $v$, and $y$ on both the up and down edges of $v$, meaning that $A$ fails around node $v$.

  Let $p'$ denote the local failure probability of $A'$. Then the local failure probability $p$ of $A$ is lower bounded by $p \geq (p'-3cf) \cdot f^3$. Again, we maximize this by setting $f = p'/(4c)$. We get that $p \geq (p' - 3cp'/(4c))(p'/(4c))^3 = (p'/4)^4 \cdot c^{-3}$. Conversely we have that $p' \leq 4p^{1/4} \cdot c^{3/4}$.
\end{proof}


\section{Lower bound in 4-regular trees} \label{sec:weak2collb}

In this section we show how Lemmas~\ref{lem:speedup1}~and~\ref{lem:speedup2} imply a lower bound of $\Omega(\log^* n)$ for weak 2-coloring. Note that, in the context of distributed complexity theory, it is usually required to obtain local high probability of success when $O(\log n)$-bit identifiers are given. In this section we prove a stronger statement, that is, at least $\frac{1}{2} \log^* n -4$ rounds are required to obtain global success probability of at least $\frac{1}{2}$, when identifiers from $1$ to $n$ are given. However, in order to prove Theorem \ref{thm:homog-complexity}, only Claim \ref{lem:prob-t-round} and Claim \ref{lem:factlocal} are required.

\begin{lemma}\label{lem:globalprob}
	Any randomized weak $2$-coloring algorithm for $4$-regular trees, with a running time less than $t=\frac{1}{2} \log^* n -b -3$ rounds, produces a legal weak 2-coloring with probability at most
	\[
	\biggl(1-\frac{1}{\log^{(2b)}n}\biggr)^{n^{\frac{1}{3(2t+1)}}} + \frac{1}{2 n^{1/3}}.
	\]
\end{lemma}
	In order to prove the lemma we proceed as follows. We consider the setting in which nodes do not have IDs, that is, the network is anonymous (as in Lemmas~\ref{lem:speedup1} and \ref{lem:speedup2}). In this setting we prove a lower bound on the failure probability of a single node (local failure probability). We can then use this result to prove a lower bound on the probability that at least one node fails (global failure probability), even in the case in which distinct IDs in $\{1,\ldots,n\}$ are given.

	\begin{claim}\label{lem:factglobal}
		If there does not exist an anonymous randomized $t$-round algorithm for weak $2$-coloring with a local failure probability at most $p$, then, even in the case in which IDs are given, there does not exist a randomized $t$-round algorithm for weak $2$-coloring with a global success probability at least
    \[
    (1-p)^{n^{\frac{1}{3(2t+1)}}} + \frac{1}{2 n^{1/3}}.
    \]
	\end{claim}
	\begin{proof}
		Consider a node $v$ in a $4$-regular tree of $n$ nodes, satisfying that $B_k(v)$, for $k = \log_3 \frac{n^{1/3}+1}{2}$, does not contain any leaf of the tree. Now $B_k(v)$ is a ball containing exactly $n^{\frac{1}{3}}$ nodes. Since the execution of a $t$-round algorithm on node $v$ depends only on nodes at distance at most $t$ from $v$, then if two nodes $u$ and $v$ are at distance at least $2t+1$, we can assume that their executions are independent. In particular, the events of nodes $u$ and $v$ failing are independent.

		In order to compute a lower bound on the global failure probability, we start by computing a lower bound on the number of independent executions in $B_k(v)$, that is, on the number of nodes that could fail independently. We give a lower bound on the number of nodes having pairwise distance at least $2t+1$ with the following process.
		Let $m(u,x)$ be the node reached by starting from $u$ and moving $2t+1$ times to the edge labeled $x$, where $x \in \{L,R,U,D\}$.
		Consider the set $I$ of nodes at distance exactly $7$ from $v$. The size of $I$ is $4 \cdot 3^6$. Start from each node $u \in I$ and consider the $3$ nodes $m(u,x)$ where $x = \{L,R,U,D\} \setminus \{\ell\}$, and $\ell$ is the label of the edge incident to $u$ on the shortest path between $u$ and $v$. Let $S$ be the set of these nodes. Repeat this process from nodes in $S$ and add the new reached nodes to $S$. At each time we add $3$ new nodes for each node added at the previous step, each one at distance at least $2t+1$ from all the other nodes. The nodes added to $S$ at step $i \ge 1$ are at distance $7+ i(2t+1)$ from $v$. Thus we can repeat the process for $k' = \bigl\lfloor \frac{k-7}{2t+1} \bigr\rfloor -1$ steps and obtain a set of nodes $S$ that are pairwise at distance at least $2t+1$, and their ball is fully contained in $B_k(v)$, meaning that in $t$ rounds they do not see nodes not contained in $B_k(v)$.
		The size of $S$ is
		  \begin{align*}
		 4 \cdot 3^6 \cdot \biggl(\sum_{i=1}^{k'} 3^i\biggr) &= 4 \cdot \frac{3^7}{2} \Bigl(3^{k'} -1\Bigr) \ge 2 \cdot 3^7 \cdot 3^{k' -1} \ge  2 \cdot 3^{7+  \frac{k-7}{2t+1} -3}  = 2 \cdot 3^{4+  \frac{k-7}{2t+1}}  \\
     &= 2 \cdot  3^{4-\frac{7}{2t+1}} \cdot \Bigl(\frac{n^{1/3}+1}{2}\Bigr)^\frac{1}{2t+1} \ge 3^{4-\frac{7}{2t+1}}  \cdot n^\frac{1}{3(2t+1)} \\&\ge n^\frac{1}{3(2t+1)},
		 \end{align*}
		 assuming that $t \ge 1$. This implies that the number of independent executions of nodes in $B_k(v)$ is at least $n^\frac{1}{3(2t+1)}$. Thus, since for an algorithm to succeed it is required that at least all nodes in $B_k(v)$ succeed, and that requires that at least all nodes in $S$ succeed, the probability that the (anonymous) algorithm succeeds is at most \[(1-p)^{n^{\frac{1}{3(2t+1)}}}.\]

		We now want to compute an upper bound on the global success probability of a randomized algorithm running on nodes labeled with distinct IDs from $\{1,\ldots,n\}$. Our speedup result is shown for anonymous randomized algorithms. Such algorithms can use randomness to generate unique identifiers with high probability. We bound the success probability of any algorithm in the randomized \local{} model by comparing it to a model where the identifiers are random but always globally unique. Such an assignment is at least as good as the worst-case assignment of unique identifiers.

		 Let $X$ be the event that the algorithm succeeds, let $Y$ the event that a random ID assignment is correct, let $X'$ be the event that the algorithm succeeds for all nodes contained in $B_k(v)$ and let $Y'$ be the event that a random ID assignment is correct (meaning that it is globally unique) in $B_k(v)$. The global success probability is at most the success probability of the algorithm running in an instance where IDs have been assigned randomly and correctly. Thus, we can upper bound the global success probability as $\Pr[X|Y] \le \Pr[X'|Y'] \le \Pr[X'] + \Pr[\bar{Y'}]$. Intuitively, $Pr[X|Y]$ is the global success probability of an algorithm running on nodes labeled with randomly assigned distinct IDs, while $\Pr[X'|Y']$ is the probability that no node fails in $B_k(v)$, given that nodes of $B_k(v)$ have randomly assigned distinct IDs. Then, $\Pr[X']$ is the probability that all nodes of $B_k(v)$ succeed for an algorithm running on an anonymous network, and $\Pr[\bar{Y'}]$ is the probability that a random ID assignment is not correct for nodes in $B_k(v)$. We will use, as upper bound for $Pr[X']$, the previously computed probability that all nodes of the set $S$ succeed.

		By assigning random IDs in $\{1,\ldots,n\}$ to the nodes of $B_k(v)$, the probability that at least a pair of nodes chooses the same ID is at most \[\binom{n^{1/3}}{2}  \frac{1}{n} < \frac{1}{2 n^{1/3}}.\] Thus, if an algorithm locally fails with probability at least $p$, then the global success probability is at most \[Pr[X|Y] \le Pr[X'] + Pr[\bar{Y'}] \le (1-p)^{n^{\frac{1}{3(2t+1)}}} +  \frac{1}{2 n^{1/3}},\] even if the IDs are provided.
	\end{proof}

   \begin{claim}\label{lem:prob-t-round}
   	    Suppose that any $0$-round weak-$c_0$ coloring algorithm locally fails with probability at least $p_0$. Then, any $t$-time weak $2$-coloring algorithm locally fails with probability at least $(\frac{p_0}{5 c_0})^{5^{2t+1}}$.
   \end{claim}
 	\begin{proof}
 		Lemma \ref{lem:speedup1} and \ref{lem:speedup2} imply that if there exists a $t$-round algorithm that produces a weak $c_t$ coloring with local failure probability at most $p_t$, then there exists a $(t-1)$-round algorithm that produces a weak $c_{t-1} = 2^{4 \hat{c}_{t-1}}$ coloring, where $\hat{c}_{t-1} = 2^{2 c_t}$, such that the failure probability is at most $p_{t-1}$, where $p_{t-1} = 4 \hat{p}_{t-1}^{1/4} \cdot \hat{c}_t^{3/4}$ and $\hat{p}_{t-1} = 5p_t^{1/5} \cdot c_t^{4/5}$. Thus, by starting with a small value of $p_t$ and applying the two lemmas recursively, we obtain a small value for $p_0$. If the obtained value of $p_0$ is smaller than the assumed minimum, we obtain a contradiction for $p_t$. This implies that we can compute a lower bound for $p_t$ as a function of $p_0$ and $c_0$. We can define $p_t$ as a function of $p_{t-1}$ as follows:
 		\[p_t = \frac{(\hat{p}_{t-1}/5)^5}{c_t^4} \text{, where } \hat{p}_{t-1} = \frac{(p_{t-1}/4)^4}{\hat{c}_t^3}  \text{, } c_{t-1} = 2^{4 \hat{c}_{t-1}} \text{, and }\hat{c}_{t-1} = 2^{2 c_t}  \text{. }\]
 		It is now possible to define a new recurrence that we can use to give a lower bound for $p_t$:
 \[p_t \ge \tilde{p}_{2t+1} \text{, where }  \tilde{p}_{i+1} =  \frac{(\tilde{p}_{i}/5)^5}{\tilde{c}_{i+1}^4} \text{, } \tilde{c}_{i} = 2^{2 \tilde{c}_{i+1}} \text{, } \tilde{c}_1 = c_0 \text{, and } \tilde{p}_1 = p_0\text{. }\]
 A lower bound for $\tilde{p}_{i+1}$ can be given by noting that $\tilde{c}_1 \ge \tilde{c}_i, ~ \forall i$. Thus,
 \[\tilde{p}_{i+1} = \frac{(\tilde{p}_{i}/5)^5}{\tilde{c}_{i+1}^4} \ge \frac{(\tilde{p}_{i}/5)^5}{\tilde{c}_{1}^4} \ge \biggl(\frac{\tilde{p}_{i}}{5 \tilde{c}_{1}}\biggr)^5.\]
 We prove the following by induction on $i$:
 \[
 \tilde{p}_{i+1}  \ge \frac{(\tilde{p}_1)^{5^{i+1}}}{(5 \tilde{c}_1)^{\sum_{j=1}^{i} 5^j}} \ge \biggl(\frac{\tilde{p}_1}{5 \tilde{c}_1}\biggr)^{5^{i+1}}.
 \]
 For $i=0$ the claim trivially holds.
 By inductive hypothesis,
 \[
 \tilde{p}_{i+1}  \ge \frac{(\tilde{p}_1)^{5^{i+1}}}{(5 \tilde{c}_1)^{\sum_{j=1}^{i} 5^j}}.
 \]
 Thus, $\tilde{p}_{i+2}$ can be lower bounded as follows:
 \[
 \tilde{p}_{i+2} \ge \biggl(\frac{\tilde{p}_{i+1}}{5 \tilde{c}_{1}}\biggr)^5 \ge \left(\frac{(\tilde{p}_1)^{5^{i+1}}}{(5 \tilde{c}_1)^{\sum_{j=1}^{i} 5^j}}\right)^5 / (5 \tilde{c}_{1})^5 = \frac{(\tilde{p}_1)^{5^{i+2}}}{(5 \tilde{c}_1)^{5 \sum_{j=1}^{i} 5^j} (5 \tilde{c}_{1})^5} = \frac{(\tilde{p}_1)^{5^{i+2}}}{(5 \tilde{c}_1)^{\sum_{j=1}^{i+1} 5^j}}.
 \]
 Finally, we obtain a lower bound for $p_t$:
 \[
 p_t \ge \tilde{p}_{2t+1} \ge \biggl(\frac{\tilde{p}_1}{5 \tilde{c}_1}\biggl)^{5^{2t+1}} = \biggl(\frac{p_0}{5 c_0}\biggr)^{5^{2t+1}}.\qedhere\]
\end{proof}

 	\begin{claim}\label{lem:factlocal}
 		Any $(\frac{1}{2} \log^* n -b -3)$-round weak $2$-coloring algorithm fails locally with probability at least $\frac{1}{\log^{(2b)}n}$, where $b \ge 1$.
 	\end{claim}
 	\begin{proof}
 		Since a $0$-round algorithm has no information about its neighbors, the random output based on the uniform distribution minimizes the failure probability, that is, the probability that a node has the same color of all its neighbors. Thus, $p_0 \ge \frac{1}{c_0^4}$. Let us fix $t=\frac{1}{2} \log^* n -b -3$. We can apply Claim \ref{lem:prob-t-round}, obtaining that the local failure probability of a $t$-round weak $2$-coloring algorithm must be at least
 		\[
    \biggl(\frac{p_0}{5 c_0}\biggr)^{5^{2t+1}} \ge \biggl(\frac{1}{5 c_0^5}\biggr)^{5^{2t+1}} \ge \biggl(\frac{1}{5 c_0}\biggr)^{5^{2t+2}}.
    \]

 		A $t$-round algorithm must produce a weak $2$-coloring, thus $c_t = 2$. Also, by definition,
    \[
    c_0 \le 2^{4\cdot 2^{\cdot^{\cdot^{\cdot^{4 \cdot 2^{4 c_t}}}}}}, \text{ where the number of ``$4\cdot 2$'' in the tower is  } 2t-1 = \log^*n -2b -7.
    \]
    In order to compute a lower bound on the failure probability $p_t$, we start by giving an upper bound for $c_0$. It is easy to see that
    \[
    \log^{(x)} c_0 \le 6 \cdot 2^{4 \cdot 2^{\cdot^{\cdot^{\cdot^{4 \cdot 2^{4 c_t}}}}}}, \text{ where the number of ``$4\cdot 2$'' in the tower is } \log^*n -2b -7 -x,
    \]
    implying that $\log^{(\log^*n -2b -5)} c_0 \le 6$. By taking the logarithm on both sides for $3$ additional times, we get that $\log^{(\log^*n -2b-2)} c_0 \le 1$. By definition, $\log^{(\log^*n -1)} n > 1$. This is equivalent to taking the logarithm $(2b+1)$ times, and then applying the logarithm again for $(\log^*n -2b-2)$ times on the result.
 		Thus,  $\log^{(\log^*n -2b-2)} \log^{(2b+1)}n > 1$. We can finally obtain a bound on $c_0$ as a function of $n$, since the above constraints imply that $c_0 \le \log^{(2b +1 )} n$.

         We can now lower bound $p_t$ by
         \[
         \biggl(\frac{1}{5 \log^{(2b +1 )} n}\biggr)^{5^{2t+2}}.
         \]
         By taking the inverse of each side we get that
         \[
         \frac{1}{p_t} \le (5\log^{(2b +1 )} n)^{ 5^{2t+2}},
         \]
         and by taking the logarithm of each side we obtain that
         \[
         \log \biggl(\frac{1}{p_t}\biggr) \le 5^{2t+2}(\log 5 + \log(\log^{(2b +1 )} n)) = 5^{2t+2}(\log 5 + \log^{(2b +2 )} n).
         \]
         For large enough $n$, this is at most $5^{2t+2}(2 \log^{(2b +2 )} n)$, which, since $t = O(\log^*n)$, can be upper bounded by $\log^{(2b +1 )} n$. Thus, we get that $p_t \ge \frac{1}{\log^{(2b)}n}$.
\end{proof}
	Now Lemma \ref{lem:globalprob} follows from Claim \ref{lem:factglobal} and Claim \ref{lem:factlocal}. We can apply it to prove Theorem~\ref{thm:main}.

\begin{theorem} \label{thm:main}
Solving weak 2-coloring with global success probability at least $\frac{1}{2}$ requires $\Omega(\log^* n)$ rounds.
\end{theorem}
\begin{proof}
		Let $t=\frac{1}{2} \log^* n -4$. Applying Lemma~\ref{lem:globalprob} with $b=1$, we get that the global success probability is at most
    \[
    \biggl(1-\frac{1}{\log \log n}\biggr)^{n^{1/(3 (2 t+1))}} + \frac{1}{2 n^{1/3}} \le e^{-\frac{n^{1/(3 (2 t+1))}}{\log \log n}} + \frac{1}{2 n^{1/3}}.
    \]
    First, note that, for large enough $n$, we have that $\frac{1}{2 n^{1/3}} < \frac{1}{4}$. In order to prove that the other addendum is also at most $\frac{1}{4},$ we prove that
    \[
    \frac{n^{1/(3 (2 t+1))}}{\log \log n} > 2.
    \]
    The above follows by noting that, since $t = O(\log^* n)$, for large enough $n$, \[\frac{\log n }{3(2t+1)} - \log \log \log n > 1.\qedhere\]
\end{proof}

\section{Generalizing the lower bound beyond 4-regular trees}\label{sec:generalization}
In this section we will describe how to adapt the presented techniques in order to prove an $\Omega(\log^* n)$ lower bound for weak $2$-coloring any $\Delta$-regular oriented tree, for any constant even $\Delta = 2k$. We start by describing how to adapt the speedup result of Section \ref{sec:speedup}.

In Section \ref{sec:speedup}, we assumed to have a consistent orientation on the two directions, horizontal and vertical. We now assume to have a consistent orientation among each of the $k$ dimensions. The new weak edge $c$-coloring problem is defined as follows. For each node $v$, there must exist one dimension $d \in \{1,\ldots,k\}$, such that the edges of dimension $d$ incident to $v$ are labeled with two different colors in $\{1,\ldots, c\}$. It is easy to see that for $k=2$ we obtain the original problem.

Then, using similar techniques of Section~\ref{sec:speedup}, it is possible to prove the following.

\begin{lemma}[First generalized speedup lemma]
	Let $A$ be a $t$-round weak $c$-coloring algorithm with local failure probability $\leq p$. Then there exists a $(t-1)$-round edge based algorithm $A'$ for weak edge $2^{2 c}$-coloring with local failure probability $p' \leq (\Delta+1) p^{1/(\Delta+1)} \cdot c^{1-1/(\Delta+1)}$.
\end{lemma}
To prove this lemma we can define the new coloring in the same way it has been defined in the $\Delta=4$ case (by computing the color probability of each endpoint of an edge). Then, we can prove that if the new algorithm fails with probability $p'$, then the old algorithm must fail with probability $p \ge (p'-\Delta c f)\cdot f^\Delta$.

\begin{lemma}[Second generalized speedup lemma]
	Let $A$ be a $t$-round weak edge $c$-coloring algorithm with local failure probability $\leq p$. Then there exists a $t$-round algorithm $A'$ for weak $2^{\Delta c}$-coloring with local failure probability $p' \leq \Delta p^{1/\Delta} \cdot c^{1-1/\Delta}$.
\end{lemma}
To prove this lemma we can define the new coloring in a similar way it has been defined in the $\Delta=4$ case. Since now a node has $\Delta$ incident edges, we obtaining a bit string of length $\Delta c$, that is a $2^{\Delta c}$ coloring. Then, we can argue that the neighborhood of a node $v'$ is good with probability at least $1 - (\Delta-1)c f$, and that the same neighborhood will be bad for a specific neighbor $v$ of $v'$ with probability $f^{\Delta-1}$, obtaining that if the new algorithm fails with probability $p'$, then the old algorithm must fail with probability $p \ge (p'- (\Delta-1) c f)\cdot f^{\Delta-1}$. This time, the choice of $f$ that maximizes the failure probability is $f=\frac{p'}{\Delta c}$. We get that \[p \ge \Bigl(\frac{p'}{\Delta}\Bigr)^\Delta \cdot \frac{1}{c^{\Delta-1}}.\]

Once we have these lemmas, we can use similar techniques of Section \ref{sec:weak2collb} to prove a lower bound for the global failure probability. First, the exact statement of Claim \ref{lem:factglobal} can be proved for any constant $\Delta$, by setting \[k= \log_{\Delta-1} \bigg( (n^{1/3} -1) \frac{\Delta-2}{\Delta}+1\bigg).\] Then, we can prove a generalized version of Claim~\ref{lem:prob-t-round} by following exactly the same reasoning.
   \begin{claim}
	Suppose that any $0$-round weak-$c_0$ coloring algorithm locally fails with probability at least $p_0$. Then, any $t$-time weak $2$-coloring algorithm locally fails with probability at least $\bigl(\frac{p_0}{(\Delta+1) c_0}\bigr)^{(\Delta+1)^{2t+1}}$.
\end{claim}
Then, using similar techniques to Claim~\ref{lem:factlocal}, we can prove that there exists some constant $b$ that depends on $\Delta$, such that any $(\frac{1}{2} \log^*n -b)$-round weak 2-coloring algorithm locally fails with probability at least $\frac{1}{\log \log n}$. Finally, we can prove the following theorem, using the same ideas of Theorem \ref{thm:main}.

\generalizedresult*


\section{Proving properties of homogeneous \lcl{}s} \label{sec:homog-proofs}

Now that we have a lower bound for the distributed computational complexity of weak $2$-coloring, we return to the topic of homogeneous \lcl{}s that we introduced in Section~\ref{sec:minimal}.

\subsection{\boldmath Solving \texorpdfstring{$P^*$}{P*} in irregular neighborhoods} \label{ssec:homog-solving}

\lemhomoglocal*

\begin{proof}
  We give an algorithm to construct the required partial labeling. Each node $v$ looks at distance $r$ and determines the closest irregularity $\irr(v)$ to it, if any. Recall that distance to a cycle is defined as sum of the minimum distance to a node on the cycle and the length of the cycle. Cycles with nodes of smaller degree are not considered as irregularities. Nodes prefer the closest cycle with the smallest maximum identifier, and if there are no cycles, the node closest with the smallest degree and smallest identifier, in that order. All nodes on the shortest path from $v$ to its closest irregularity also have an irregularity within distance $r$.

  First assume that the closest irregularity to node $v$ is a cycle. There are two cases: either $v$ is on the cycle, or not. In the latter case, $v$ points toward the cycle and outputs $d(v) = 0$. As all nodes on the path to the cycle do this, we will have that $d(v) = d(p(v))$. Now assume $v$ is on a cycle. The algorithm must ensure that the cycle is oriented and labeled in a consistent manner.

  Since each node $v$ on $C$ will have a cycle as their closest irregularity, they will all output $d(v) = 0$. For each cycle $C$, a consistent orientation is determined as follows: the node with the smallest identifier on $C$ will orient toward its smaller neighbor, and all nodes on $C$ follow this orientation. Each node $v$ will determine the orientation of $C = \irr(v)$ and point $p(v)$ according to the orientation of $C$.

  If for all $u \in C$ we have that $\irr(u) = C$, this orientation is consistent and correct. Now assume that there exists $u \in C$ such that $\irr(u) = C' \neq C$. If $\irr(p(v)) = C$, the labeling is still correct around $v$. If $\irr(p(v)) \neq C$, it still holds that $p(p(v)) \neq v$, since $p(p(v)) = v$ would imply that $\irr(v) \neq C$, and the labeling is correct.

  Finally, assume that the closest irregularity to node $v$ is a node $u$ of degree $\deg(u) < \Delta$. Since nodes on the shortest path $P$ from $v$ to $u$ might have a cycle as their closest irregularity, $v$ must look an additional $r$ steps away and verify if this is the case. If not, $v$ will point along $P$ and output $d(v) = \deg(u)$. If yes, $v$ will still point along $P$, but output $d(v) = 0$. Let $w$ be the first node on $P$ from $v$ such that $\irr(w)$ is some cycle $C$. If $p(v) \neq w$, the labeling is correct around $v$. If $p(v) = w$, we have that $d(w) = 0$ and $p(w) \neq v$ since otherwise $v$ would have a cycle as its irregularity. Therefore the labeling is correct in the neighborhood of $v$.
\end{proof}

\subsection{\boldmath Complexity of \texorpdfstring{$P^*$}{P*}} \label{ssec:homog-complexity}

In this section we prove the two lemmas to show that the distributed complexity of $P^*$, defined in Section~\ref{sec:weak-is-minimal} is $\Theta(\log n)$.

\thmsymcomplexity*

Theorem~\ref{thm:symcomplexity} follows directly from the following two lemmas.

\begin{lemma} \label{lem:symupper}
  Problem $P^*$ is in $\dlocal(O(\log n))$.
\end{lemma}

\begin{proof}
  Let $G = (V,E)$ be the input graph of size $n$ and maximum degree $\Delta$. Any node $v$ in $G$ will see a node of degree $\leq \Delta-1$ or a cycle within distance $O(\log_{\Delta} n)$. Applying Lemma~\ref{lem:homog-local}, we can solve $P^*$ in time $O(\log_{\Delta} n)$.
\end{proof}

\begin{lemma} \label{lem:symlower}
  Solving $P^*$ in the randomized \local{} model requires time $\Omega(\log_{\Delta} n)$.
\end{lemma}

\begin{proof}
  We consider two trees, $T$ and $T'$. The first is a balanced $\Delta$-regular tree, with $\Delta > 2$: the center vertex $v$ is at distance $r = \Theta(\log_{\Delta} n)$ from each leaf. Assume we have an algorithm $A$ with running time $t(n) = o(\log_{\Delta} n)$ that solves $P^*$ on $T$. For a large enough $n$, we have that $t(n) < r - 1$.

  Now no matter which neighbor $v$ points to, it must output $d(v) = 1$. Now consider the following $T'$: it is exactly as $T$, except for each node $u$ at distance $r-1$ from $v$. For each of these nodes, we remove one of its neighboring leaves and place it as a neighbor of one of the other leaves. Since $|V(T)| = |V(T')| = n$, we have that $A$ must also run in time $t(n) < r-1$ on $T'$. Therefore the $t(n)$-neighborhood of $v$ is indistinguishable between $T$ and $T'$.

  Each node $u$ at distance $r-1$ from $v$ (and every irregularity reachable from $v$ via a path of nodes of degree $\Delta$) has degree $\deg(u) = \Delta-1$ and therefore must output $p(u) = \bot$ and $d(u) = \Delta-1$. The pointer chain starting from $v$ can only terminate at the nodes at distance $r-1$, and therefore $v$ must also output $d(v) = \Delta-1$ on $T'$. Therefore the algorithm must fail on either $T$ or $T'$ with probability at least $1/2$.
\end{proof}

\subsection{\boldmath Classification of homogeneous \lcl{}s} \label{ssec:homog-classfication}

Finally, we are ready to prove Theorem~\ref{thm:homog-complexity}.

\thmhomogcomplexity*

\begin{proof}
  If the randomized complexity of an \lcl{} is at most logarithmic, it can fall into one of the following four categories: $O(1)$, between $\Omega(\log \log^* n)$ and $O(\log^* n)$, $\Theta(\log \log n)$, and $\Theta(\log n)$; see Appendix~\ref{sec:gap-appendix} for details. Our main result implies that homogeneous \lcl{}s cannot have a complexity between $\omega(1)$ and $o(\log^* n)$, as we show below.

  Let $P_H = P \cup P^*$ be a $\Delta$-homogeneous \lcl{}. Let $P$ be locally checkable in $r$ rounds. We divide all \lcl{}s into two types based on whether some constant label is valid for $P$ in the $\Delta$-regular tree, or not.

\begin{enumerate}
  \item Constant label is valid inside $\Delta$-regular trees. Each node outputs that constant label as solution to $P$. If its $r$-neighborhood is a tree, this is a valid output. If not, there is an irregularity within distance $r$. Using Lemma~\ref{lem:homog-local}, it is possible to solve $P^*$ in time $O(r)$. Therefore we have that $P_H \in \dlocal(O(1))$.

  \item Constant label is not valid inside $\Delta$-regular trees. Let $A$ be an algorithm for $P_H$. In each $r$-neighborhood $A$ must solve $P$ or $P^*$. On $\Delta$-regular trees the latter requires $\Omega(\log n)$ rounds. Therefore assume that on trees $A$ solves $P$. Since a constant label is not a feasible output in this case, each node $v$ is guaranteed to have a node $u$ inside its radius-$r$ neighborhood such that $A(v) \neq A(u)$. This implies that the output of $A$ forms a distance-$r$ weak coloring with a constant number of colors. By Lemma~\ref{lem:weak-col-family} this gives a weak 2-coloring in a constant number of rounds. By Theorem~\ref{thm:main}, this requires $\Omega(\log^* n)$ rounds, so we have that $P$, and also $P_H$, require $\Omega(\log^* n)$ rounds. If $P \in \rlocal(O(\log^* n))$, then by the derandomization result of Chang et al.\ \cite{chang16exponential} (see Theorem~\ref{thm:derandomization} in Appendix~\ref{sec:gap-appendix}) we also have that $P \in \dlocal(O(\log^* n))$.

  \item $P \in \rlocal(o(\log n)) \setminus \rlocal(O(\log^* n))$. By Theorem~\ref{thm:sublog-trees}, due to Chang and Pettie~\cite{chang17hierarchy}, and Chang et al.\ \cite{chang18complexity}, there is a randomized algorithm $A$ for solving $P$ in time $O(\log \log n)$ in neighborhoods that are tree-like up to distance $\Omega(\log \log n)$.
  In neighborhoods that are not tree-like up to that distance, $P^*$ can be solved in time $o(\log \log n)$. Let $T(n)$ denote the running time of the randomized algorithm to solve $P$ on trees. Run the following algorithm: in parallel try to solve $P$ with the $T(n)$-time distributed algorithm, and try to solve $P^*$ for parameter $k = T(n)+r$, as in Lemma~\ref{lem:homog-local}. Every $k$-neighborhood of a node $v$ is either a tree, in which case $A$ succeeds in solving $P$ in the $r$-neighborhood of $v$, or there is an irregularity in $B_k(v)$, in which case all nodes in $B_1(v)$ solve $P^*$ in time $O(k)$.

  Since deterministic complexity cannot be between $\omega(\log^* n)$ and $o(\log n)$ by Theorem~\ref{thm:det-logstar-log-gap}, it must be exactly $\Theta(\log n)$.

  \item $P \notin \rlocal(o(\log n))$. Finally, even if $P$ has randomized complexity $\omega(\log n)$, we can solve $P^*$ in time $O(\log n)$. \qedhere
\end{enumerate}
\end{proof}

\section*{Acknowledgments}
We would like to thank Sebastian Brandt for discussions related to weak 2-coloring, Tuomo Lempi\"ainen for discussions related to the concept of minimality, and anonymous reviewers for their helpful comments on previous versions of this work.
This work was supported in part by the Academy of Finland, Grants 285721 and 314888.

\newpage

{
	\urlstyle{sf}
	\DeclareUrlCommand{\Doi}{\urlstyle{same}}
	\renewcommand{\doi}[1]{\href{http://dx.doi.org/#1}{\footnotesize\sf doi:\Doi{#1}}}
	\bibliographystyle{plainnat}
	\bibliography{weak2col}
}

\appendix

\newpage

\section{Appendix}

\subsection{\boldmath Gaps in the complexity landscape of \lcl{} problems}\label{sec:gap-appendix}

In this section, we review the theorems that are necessary to prove all the gaps required to fully characterize the possible complexities of homogeneous \lcl{}s. The proofs give a sketch of how exactly these gaps can be shown. We emphasize that all of these follow from previous work.

The following theorem connects the randomized and deterministic complexities of \lcl{}s.

\begin{theorem}[\cite{chang16exponential}, Theorem~3]\label{thm:derandomization}
  The deterministic complexity of any \lcl{} $P$ on instances of size $n$ is upper bounded by the randomized complexity of $P$ on instances of size $2^{n^2}$.
\end{theorem}

\begin{corollary}[\cite{Naor1995,chang16exponential}] \label{cor:sublogstar-rand-det}
  For \lcl{}s, and for each $f(n) = O(\log^* n)$, we have that $\rlocal{}(f(n)) \subseteq \dlocal{}(O(f(n)))$.
\end{corollary}

\begin{proof}
  For $f(n) = O(1)$, this was shown by Naor and Stockmeyer \cite{Naor1995}. The generalization follows from Theorem~\ref{thm:derandomization}.
\end{proof}

\begin{theorem}[\cite{Naor1995,chang17hierarchy}] \label{thm:const-loglogstar-gap}
  There are no \lcl{}s with complexity $\omega(1)$ and $o(\log \log^* n)$.
\end{theorem}

\begin{proof}
    Naor and Stockmeyer~\cite{Naor1995} used Ramsey's theorem to argue that a constant-time deterministic algorithm for an \lcl{} implies a constant-time algorithm that only uses the relative order of the node identifiers. A more careful analysis implies that this analysis extends to $o(\log^* n)$ time on rings and $o(\log \log^* n)$ time on bounded-degree graphs (see \cite{chang17hierarchy}, Appendix A for a detailed analysis).

    An order-invariant algorithm cannot run in non-constant time that is $o(\log n)$, since the input does not reveal anything about the size of the graph before cycles close: therefore this implies a gap between deterministic running times $\omega(1)$ and $o(\log \log^* n)$. By Corollary~\ref{cor:sublogstar-rand-det} this gap extends to randomized algorithms.
\end{proof}

\begin{theorem}[\cite{chang16exponential}] \label{thm:det-logstar-log-gap}
  There are no \lcl{} problems with deterministic complexity between $\omega(\log^* n)$ and $o(\log n)$.
\end{theorem}

\begin{proof}
    This follows directly from Theorem~6 of Chang et al.\ \cite{chang16exponential}: given a $T(n) = o(\log n)$-round algorithm for an \lcl{} $P$, we can lie to the algorithm that the graph is of much smaller size and compute a coloring of the nodes (in time $O(\log^* n)$) that looks locally like an identifier setting in a small graph. Since the algorithm is fooled everywhere locally, it must produce a proper solution everywhere locally. Since we are solving an \lcl{}, this is a globally feasible solution.
\end{proof}

\begin{theorem}[\cite{chang16exponential}] \label{thm:rand-logstar-loglog-gap}
  There are no \lcl{} problems with randomized complexity between $\omega(\log^* n)$ and $o(\log \log n)$.
\end{theorem}

\begin{proof}
    This follows from Theorem~\ref{thm:det-logstar-log-gap} and Theorem~\ref{thm:derandomization}, both due to Chang et al.\ \cite{chang16exponential}.
    Theorem~\ref{thm:derandomization} states that the randomized complexity of an \lcl{} $P$ on instances of size $n$ is at least the deterministic complexity of $P$ on instances of size $2^{n^2}$.
    Since, by Theorem~\ref{thm:det-logstar-log-gap} there are no \lcl{}s with deterministic complexity $\omega(\log^* n)$ and $o(\log n)$, there can be no \lcl{}s with randomized complexity between $\omega(\log^* n)$ and $o(\log \log n)$.
\end{proof}

\begin{theorem}[\cite{chang17hierarchy,chang18complexity}] \label{thm:sublog-trees}
  On trees, all \lcl{}s that can be solved in time $o(\log n)$ can be solved in time $O(\log \log n)$.
\end{theorem}

\begin{proof}
    This follows from the fact that the distributed Lov\'{a}sz local lemma is complete for \lcl{}s in sublogarithmic time~\cite{chang17hierarchy}, and that there exists an $O(\log \log n)$-time algorithm for solving LLL on tree-structured instances~\cite{chang18complexity}.

    In particular, any sublogarithmic-time randomized algorithm for an \lcl{} yields an instance of (symmetric) LLL with a polynomial LLL criterion $pd^c < 1$, for an arbitrarily large constant $c$. On trees, this instance is tree structured: the dependency graph looks locally like some $r$th power of a tree.

    Chang et al.~\cite{chang18complexity} give a randomized algorithm for solving tree-structured LLLs with criterion $p(ed)^{\lambda}$ in time $O(\log_{\lambda} \log n)$.
\end{proof}

\end{document}